\documentclass{llncs}
\usepackage{graphicx}

\usepackage{cite,latexsym,graphicx,wrapfig,fancyhdr,url,subfigure,times,mathptm}

\DeclareSymbolFont{AMSb}{U}{msb}{m}{n}
\DeclareSymbolFontAlphabet{\Bbb}{AMSb}

\def\Z{\ensuremath{\Bbb Z}}
\def\bot{{-\infty}}	
\def\top{{+\infty}}

\newcommand{\graph}{\mathcal{G}}
\newcommand{\layout}{\mathcal{L}}

\newcommand{\porder}{\mathcal{P}}
\newcommand{\qorder}{\mathcal{Q}}
\newcommand {\REL} {regular edge labeling}

\newtheorem{cor}{Corollary}

\graphicspath{{figures/}}

\title{Orientation-Constrained Rectangular Layouts}

\author{David Eppstein \inst{1} \and Elena Mumford \inst{2}}

\institute{Department of Computer Science, University of California, Irvine, USA
\and
Department of Mathematics and Computer Science, TU
Eindhoven, The Netherlands
}

\begin{document}
\maketitle

\begin{abstract}
We construct partitions of rectangles into smaller rectangles from an input consisting of a planar dual graph of the layout together with restrictions on the orientations of edges and junctions of the layout. Such an orientation-constrained layout, if it exists, may be constructed in polynomial time, and all orientation-constrained layouts may be listed in polynomial time per layout.
\end{abstract}


\section{Introduction}\label{sec:def}
Consider a partition of a rectangle into smaller rectangles, at most three of which meet at any point. We call such a partition a \emph{rectangular layout}.
Rectangular layouts are an important tool in many application areas. In VLSI design rectangular layouts represent floorplans of integrated circuits~\cite{LiaLuYen-Algs-03}, while in architectural design they represent floorplans of buildings~\cite{Earl1979,Rinsma1988}. In cartography they are used to visualize numeric data about geographic regions, by stylizing the shapes of a set of regions to become rectangles, with areas chosen to represent statistical data about the regions; such visualizations are called \emph{rectangular cartograms}, and were first introduced in 1934 by Raisz~\cite{Rai-GR-34}.

The \emph{dual graph} or \emph{adjacency graph} of a layout is a plane graph $\graph(\layout)$ that has a vertex for every region of $\layout$ and an edge for every two adjacent regions. In both VLSI design and in cartogram construction, the adjacency graph $\graph$ is typically given as input, and one has to construct its \emph{rectangular dual}, a rectangular layout for which $\graph$ is the dual graph. Necessary and sufficient conditions for a graph to have a rectangular dual are known~\cite{KozKin-Nw-85}, but
graphs that admit a rectangular dual often admit more than one. This fact allows us to impose additional requirements on the rectangular duals that we select, but it also leads to difficult algorithmic questions concerning problems of finding layouts with desired properties. For example, Eppstein et al~\cite{area-uni} have considered the search for \emph{area-universal} layouts, layouts that can be turned into rectangular cartograms for any assignment of positive weights to their regions.

In this paper, we consider another kind of constrained layouts. Given a graph $\graph$ we would like to know whether $\graph$ has a rectangular dual that satisfies certain constraints on the orientations of the adjacencies of its regions; such constraints may be particularly relevant for cartographic applications of these layouts. For example, in a cartogram of the U.S., we might require that a rectangle representing Nevada be right of or above a rectangle representing California, as geographically Nevada is east and north of California. We show that layouts with orientation constraints of this type may be constructed in polynomial time.  Further, we can list all layouts obeying the constraints in polynomial time per layout. Our algorithms can handle constraints (such as the one above) limiting the allowed orientations of a shared edge between a pair of adjacent regions, as well as more general kind of constraints restricting the possible orientations of the three rectangles meeting at  any junction of the layout.  We also discuss the problem of finding area-universal layouts in the presence of constraints of these types. A version of the orientation-restricted layout problem was previously considered by van Kreveld and Speckmann~\cite{KreSpe-CGTA-07} but they required a more restrictive set of constraints and searched exhaustively through all layouts rather than developing polynomial time algorithms.

Following \cite{area-uni}, we use Birkhoff's representation theorem for finite distributive lattices to associate the layouts dual to $\graph$ with partitions of a related partial order into a lower set and an upper set. The main idea of our new algorithms is to translate the orientation constraints of our problem into an equivalence relation on this partial order. We form a quasiorder by combining this relation with the original partial order, partition the quasiorder into lower and upper sets, and construct layouts from these partitions. However, the theory as outlined above only works directly on dual graphs with no nontrivial separating 4-cycles. To handle the general case we must do more work to partition $\graph$ by its 4-cycles into subgraphs and to piece together the solutions from each subgraph.

\section{Preliminaries}

Kozminski and Kinnen~\cite{KozKin-Nw-85} demonstrated that a plane triangulated graph $\graph$ has a rectangular dual if it can be augmented with four external vertices $\{l,t,r,b\}$ to obtain an \emph{extended graph} $E(\graph)$ in which every inner face is a triangle, the outer face is a quadrilateral, and $E(\graph)$ does not contain any separating 3-cycles (a \emph{separating $k$-cycle} is a $k$-cycle that has vertices both inside and outside of it). A graph $\graph$ that can be extended in this way is said to be \emph{proper}---see Fig.~\ref{fig:rect-dual} for an example.
\begin{figure}[t]
  \centering
  \includegraphics[scale=0.8]{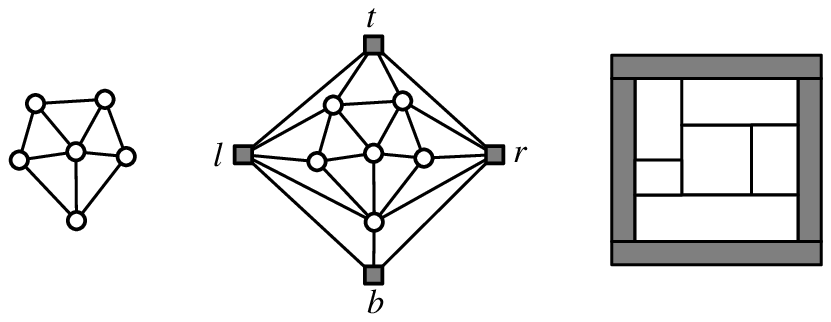}
  \caption{A proper graph $\graph$, extended graph $E(\graph)$, and rectangular dual $\layout$ of $E(\graph)$, from~\cite{area-uni}.}
\label{fig:rect-dual}
\end{figure}
The extended graph $E(\graph)$ is sometimes referred to as a \emph{corner assignment} of $\graph$, since it defines which vertices of $\graph$ become corner rectangles of the corresponding dual layout.  In the rest of the paper we assume that we are given a proper graph with a corner assignment. For proper graphs without a corner assignment, one can always test all possible corner assignments, as their number is polynomial in the number of external vertices of the graph.

A rectangular dual $\layout$ induces a labeling for the edges of its graph $\graph(\layout)$: we color each edge blue if the corresponding pair of rectangles share a vertical line, or red if the corresponding pair of rectangles share a horizontal border; we direct the blue edges from left to right and red edges from bottom to top. For each inner vertex $v$ of $\graph(\layout)$ the incident edges with the same label form continuous blocks around $v$: all incoming blue edges are followed (in clockwise order) by all outgoing red, all  outgoing blue and finally all incoming red edges.  All edges adjacent to one of the four external vertices $l,t,r,b$ have a single label.
A labeling that satisfies these properties is called a \emph{regular edge labeling}~\cite{KanHe-TCS-97}. Each regular edge labeling of a proper graph corresponds to an equivalence class of rectangular duals  of $\graph$, considering two duals to be \emph{equivalent} whenever every pair of adjacent regions have the same type of adjacency in both duals.

\subsection{The distributive lattice of  {\REL}s}

A \emph{distributive lattice} is a partially ordered set in which every pair $(a,b)$ of elements has a unique supremum $a \wedge b$ (called the \emph{meet} of $a$ and $b$) and a unique infinum $a \vee b$ (called the \emph{join} of $a$ and $b$) and where the join and meet operations are distributive over each other. Two comparable elements that are closest neighbours in the order are said to be a \emph{covering pair}, the larger one is said to \emph{cover} the smaller one.

\begin{figure}[t]
\centering\includegraphics[width=.4\textwidth]{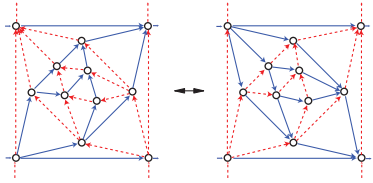}
\caption{Recoloring the interior of an alternatingly-colored four-cycle in a regular edge labeling.}
\label{fig:edge-move}
\end{figure}

All regular edge labelings of a proper graph form a distributive lattice~\cite{Fus-GD-05,Fus-DM-08,TanChe-ISCAS-90} in which the covering pairs of layouts are the pairs in which the layouts can be transformed from one to the other by means of a \emph{move}---changing the labeling of the edges inside an \emph{alternating four-cycle} (a 4-cycle $C$ whose edge colors alternate along the cycle). Each red edge within the cycle becomes blue and vice versa; the orientations of the edges are adjusted in a unique way such that the cyclic order of the edges around each vertex is as defined above---see Fig.~\ref{fig:edge-move} for an example. In terms of layouts, the move means rotating the sublayout formed by the inner vertices of $C$ by 90 degrees.  A move is called \emph{clockwise} if the borders between the red and blue labels of each of the four vertices of the cycle move clockwise by the move, and called \emph{counterclockwise} otherwise. A counterclockwise move transforms a layout into another layout higher up the lattice.

We can represent the lattice by a graph in which each vertex represents a single layout and each edge represents a move between a covering pair of layouts, directed from the lower layout to the higher one. We define a \emph{monotone path} to be a path in this graph corresponding to a sequence of counterclockwise moves.

\subsection{The Birkhoff representation of the lattice of layouts}

\begin{figure}[t]
\centering
\includegraphics[width=3.5in]{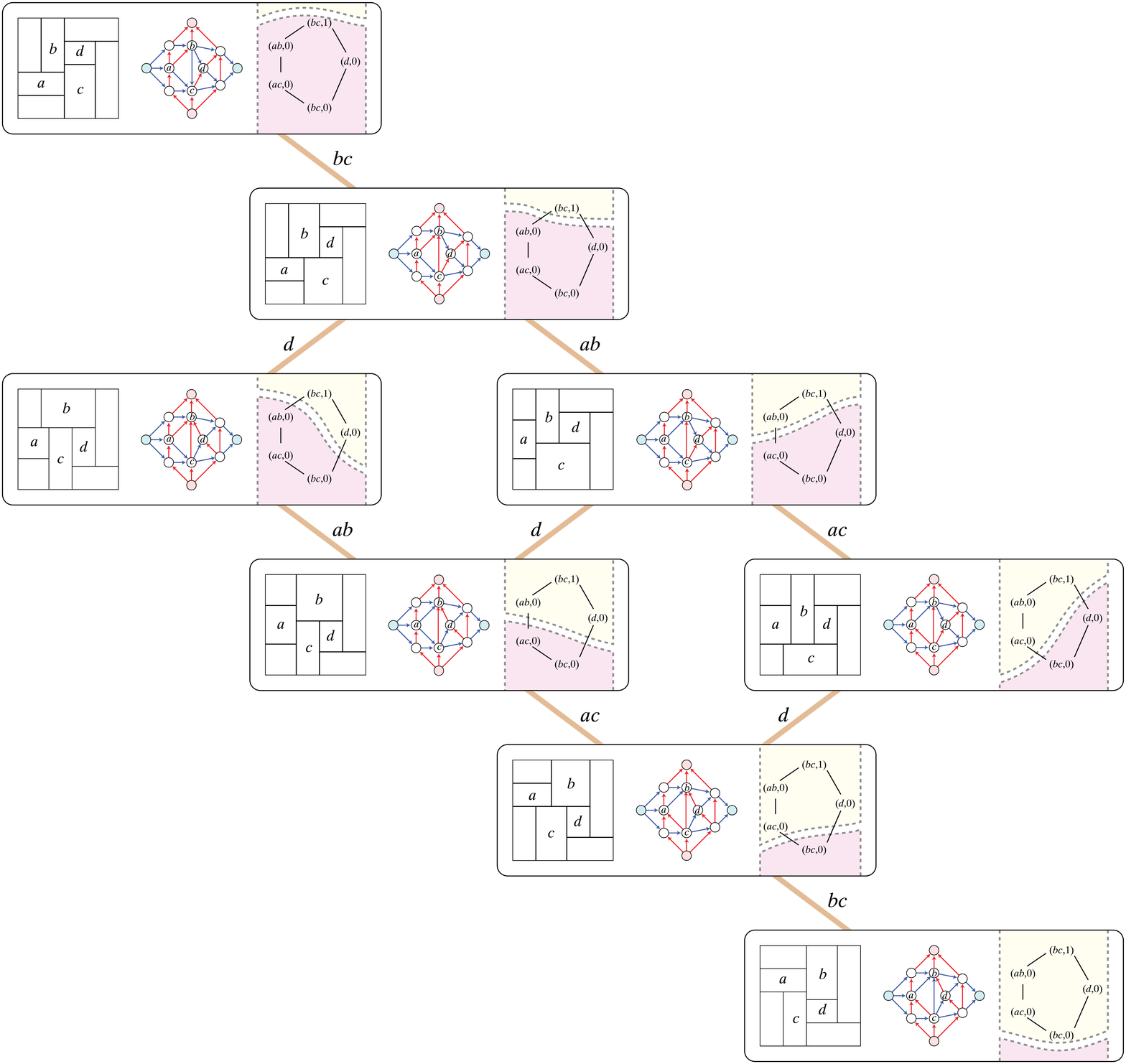}
\vspace{-0.75\baselineskip}
\caption{The rectangular layouts dual to a given extended graph $E(\graph)$ and the corresponding regular edge labelings and partial order partitions. Two layouts are shown connected to each other by an edge if they differ by reversing the color within a single alternatingly-colored four-cycle; these moves are labeled by the edge or vertex within the four-cycle. From~\cite{area-uni}.}
\label{fig:8layouts}
\end{figure}

For any finite distributive lattice $D$, let $\porder$ be the partial order of join-irreducible elements (elements that cover only one other element of $D$), and let $J(\porder)$ be the lattice of partitions of $\porder$ into sets $L$ and $U$, where $L$ is downward closed and $U$ is upward closed and where meets and joins in $J(\porder)$ are defined as intersections and unions of these sets. Birkhoff's representation theorem~\cite{Bir-DMJ-37} states that $D$ is isomorphic to $J(\porder)$.

Eppstein et al.~\cite{area-uni} show that when $E(\graph)$ has no nontrivial separating four-cycles (four-cycles with more than one vertex on the inside) the partial order of join-irreducible elements of the lattice of rectangular duals of $E(\graph)$ is order-isomorphic to the partial order $\porder$ on pairs $(x,i)$, where $x$ is a \emph{flippable item} of $E(\graph)$, and $i$ is a \emph{flipping number} of $x$.
A \emph{flippable item} $x$ is either a degree-four vertex of $\graph$ or an edge of $\graph$ that is not adjacent to a degree-four vertex, such that there exist two {\REL}s of $E(\graph)$ in which $x$ has different labels (when $x$ is a vertex we refer to the labels of its four adjacent edges).
Given a layout $\layout$ and a flippable item $x$, the number $f_x(\layout)$ is the number of times that $x$ has been flipped on a monotone path in the distributive lattice from its bottom element to $\layout$; this number, which we call the \emph{flipping number} of $x$ in $\layout$, is well defined, since it is independent of the path by which $\layout$ has been reached.
For every flippable item $x$, $\porder$ contains pairs $(x,i)$ for all $i$ such that there exist a layout  $\layout$ where $f_x(\layout) = i-1$. A pair $(x,i)$ is associated with the transition of $x$ from state $i$ to $i+1$.
A pair $(x,i)$ is less than a pair $(y,j)$ in the partial order if is not possible to flip $y$ for the $j$th time before flipping $x$ for the $i$th time. If $(x,i)$ and $(y,j)$ form a covering pair in $\porder$, the flippable items $x$ and $y$ belong to the same triangular face of $E(\graph)$.

As Eppstein et al. show, the layouts dual to $E(\graph)$ correspond one-for-one with partitions of the partial order $\porder$ into a lower set $L$ and an upper set $U$.
The labeling of the layout corresponding to a given partition of $\porder$ can be recovered by starting from the minimal layout and flipping each flippable item $x$ represented in the lower set $L$ of the partition $n_x+1$ times, where $(x, n_x)$ is the highest pair involving $x$ in $L$.
The downward moves that can be performed from $\layout$ correspond to the maximal elements of $L$, and the upward moves that can be performed from $\layout$ correspond to the minimal elements of~$U$.
Fig.~\ref{fig:8layouts} depicts the lattice of layouts of a 12-vertex extended dual graph, showing for each layout the corresponding partition of the partial order into two sets $L$ and~$U$.
The partial order of flippable items has at most $O(n^2)$ elements and can be constructed in time polynomial in $n$, where $n$ is the number of vertices in $\graph$~\cite{area-uni}.

\section{The lattice theory of constrained layouts}

As we describe in this section, in the case where every separating 4-cycle in $E(\graph)$ is trivial, the orientation-constrained layouts of $E(\graph)$ may themselves be described as a distributive lattice, a sublattice (although not in general a connected subgraph) of the lattice of all layouts of $E(\graph)$.

\subsection{Sublattices from quotient quasiorders}

We first consider a more general order-theoretic problem.
Let $\porder$ be a partial order and let $C$ be a (disconnected) undirected \emph{constraint graph} having the elements of $\porder$ as its vertices. We say that a partition of $\porder$ into a lower set $L$ and an upper set $U$ \emph{respects} $C$ if there does not exist an edge of $C$ that has one endpoint in $L$ and the other endpoint in $U$. As we now show, the partitions that respect $C$ may be described as a sublattice of the distributive lattice $J(\porder)$ defined via Birkhoff's representation theorem from $\porder$.

We define a quasiorder (that is, reflexive and transitive binary relation) $\qorder$ on the same elements as $\porder$, by adding pairs to the relation that cause certain elements of $\porder$ to become equivalent to each other. More precisely, form a directed graph that has the elements of $\porder$ as its vertices , and that has a directed edge from $x$ to $y$ whenever either $x\le y$ in $\porder$ or $xy$ is an edge in $C$, and define $\qorder$ to be the transitive closure of this directed graph: that is, $(x,y)$ is a relation in $\qorder$ whenever there is  a path from $x$ to $y$ in the directed graph. A subset $S$ of $\qorder$ is downward closed (respectively, upward closed) if there is no pair $(x,y)$ related in $Q$ for which $S\cap\{x,y\}=\{y\}$ (respectively, $S\cap\{x,y\}=\{x\}$).

Denote by $J(\qorder)$ the set of partitions of $\qorder$ into a downward closed and an upward closed set. Each strongly connected component of the directed graph derived from $\porder$ and $C$ corresponds to a set of elements of $\qorder$ that are all related bidirectionally to each other, and $\qorder$ induces a partial order on these strongly connected components.
Therefore, by Birkhoff's representation theorem, $J(\qorder)$ forms a distributive lattice under set unions and intersections.

\begin{lemma}
\label{lem:P-Q}
The family of partitions in $J(\qorder)$ is the family of partitions of $\porder$ into lower and upper sets that respect $C$.
\end{lemma}

\begin{proof}
We show the lemma by demonstrating that every partition in $J(\qorder)$ corresponds to a partition of $J(\porder)$ that respects $C$ and the other way round.

In one direction, let $(L,U)$ be a partition in $J(\qorder)$. Then, since $\qorder\supset\porder$, it follows that $(L,U)$ is also a partition of $\porder$ into a downward-closed and an upward-closed subset. Additionally, $(L,U)$ respects $C$, for if there were an edge $xy$ of $C$ with one endpoint in $L$ and the other endpoint in $U$ then one of the two pairs $(x,y)$ or $(y,x)$ would contradict the definition of being downward closed for $L$.

In the other direction, let $(L',U')$ be a partition of $\porder$ into upper and lower sets that respects $C$, let $(x,y)$ be any pair in $\qorder$, and suppose for a contradiction that $x\in U'$ and $y\in L'$. Then there exists a directed path from $x$ to $y$ in which each edge consists either of an ordered pair in $\porder$ or an edge in $C$. Since $x\in U'$ and $y\in L'$, this path must have an edge in which the first endpoint is in $U'$ and the second endpoint is in $L'$. But if this edge comes from an ordered pair in $\porder$, then $(L',U')$ is not a partition of $\porder$ into upper and lower sets, while if this edge comes from $C$ then $(L',U')$ does not respect $C$. This contradiction establishes that there can be no such pair $(x,y)$, so $(L',U')$ is a partition of $\qorder$ into upper and lower sets as we needed to establish.
\end{proof}

If $\porder$ and $C$ are given as input, we may construct $\qorder$ in polynomial time: by finding strongly connected components of $\qorder$ we may reduce it to a partial order, after which it is straightforward to list the partitions in $J(\qorder)$ in polynomial time per partition.

\subsection{Edge orientation constraints}

Consider a proper graph $\graph$ with corner assignment $E(\graph)$ and assume that each edge $e$ is given with a set of \emph{forbidden labels}, where a labels is a color-orientation combination for an edge, and let $\porder$ be the partial order whose associated distributive lattice $J(\porder)$ has its elements in one-to-one correspondence with the layouts of $E(\graph)$.
Let $x$ be the flippable item corresponding to $e$---that is either the edge itself of the degree-four vertex $e$ is adjacent to. Then in any layout $\layout$, corresponding to a partition $(L,U)\in J(\porder)$, the orientation of $e$ in $\layout$ may be determined from $i\bmod 4$, where $i$ is the largest value such that $(x,i)\in L$. Thus if we would like to exclude a certain color-orientation combination for $x$, we have find the corresponding value $k \in \Z_4$ and exclude the layouts $\layout$ such that $f_x(\layout) = k \bmod 4$  from consideration. Thus the set of flipping values for $x$ can be partitioned into \emph{forbidden} and \emph{legal} values for $x$; instead of considering color-orientation combinations of the flippable items we may consider their flipping values.
We formalize this reasoning in the following lemma.

\begin{lemma}
\label{lem:constraint-flip}
Let $E(\graph)$ be a corner assignment of a proper graph $\graph$. Let $x$ be a flippable item in $E(\graph)$, let $\layout$ be an element of the lattice of \REL s of $E(\graph)$, and let $(L,U)$ be the corresponding partition of $\porder$.

Then $\layout$  satisfies the constraints described by the forbidden labels if and only if for every flippable item $x$ one of the following is true:
\begin{itemize}
\item The highest pair involving $x$ in $L$ is $(x,i)$, where $i+1$ is not a forbidden value for $x$, or
\item $(x,0)$ is in the upper set and  $0$ is not a forbidden value for $x$.
\end{itemize}
\end{lemma}

Lemma~\ref{lem:P-Q} may be used to show that the set of all constrained layout is a distributive lattice, and that all constrained layouts may be listed in polynomial time per layout. For technical reasons we augment $\porder$ to a new partial order $A(\porder)=\porder\cup\{\bot,\top\}$, where the new element $\bot$ lies below all other elements and the new element $\top$ lies above all other elements. Each layout of $E(\graph)$ corresponds to a partition of $\porder$ into lower and upper sets, which can be mapped into a partition of $A(\porder)$ by adding $\bot$ to the lower set and $\top$ to the upper set. The distributive lattice $J(A(\porder))$ thus has two additional elements that do not correspond to layouts of $E(\graph)$: one in which the lower set is empty and one in which the upper set is empty. We define a constraint graph $C$ having as its vertices the elements of $A(\porder)$, with edges defined as follows:
\begin{itemize}
\item If $(x,i)$ and $(x,i+1)$ are both elements of $A(\porder)$ and $i+1$ is a forbidden value for $x$, we add an edge from $(x,i)$ to $(x,i+1)$ in $C$.
\item If $(x,i)$ is an element of $A(\porder)$ but $(x,i+1)$ is not, and $i+1$ is a forbidden value for $x$, we add an edge from $(x,i)$ to $\top$ in $C$.
\item If $0$ is a forbidden value for $x$, we add an edge from $\bot$ to $(x,0)$ in $C$.
\end{itemize}

\begin{figure}[t]
\centering
\includegraphics[width=4in]{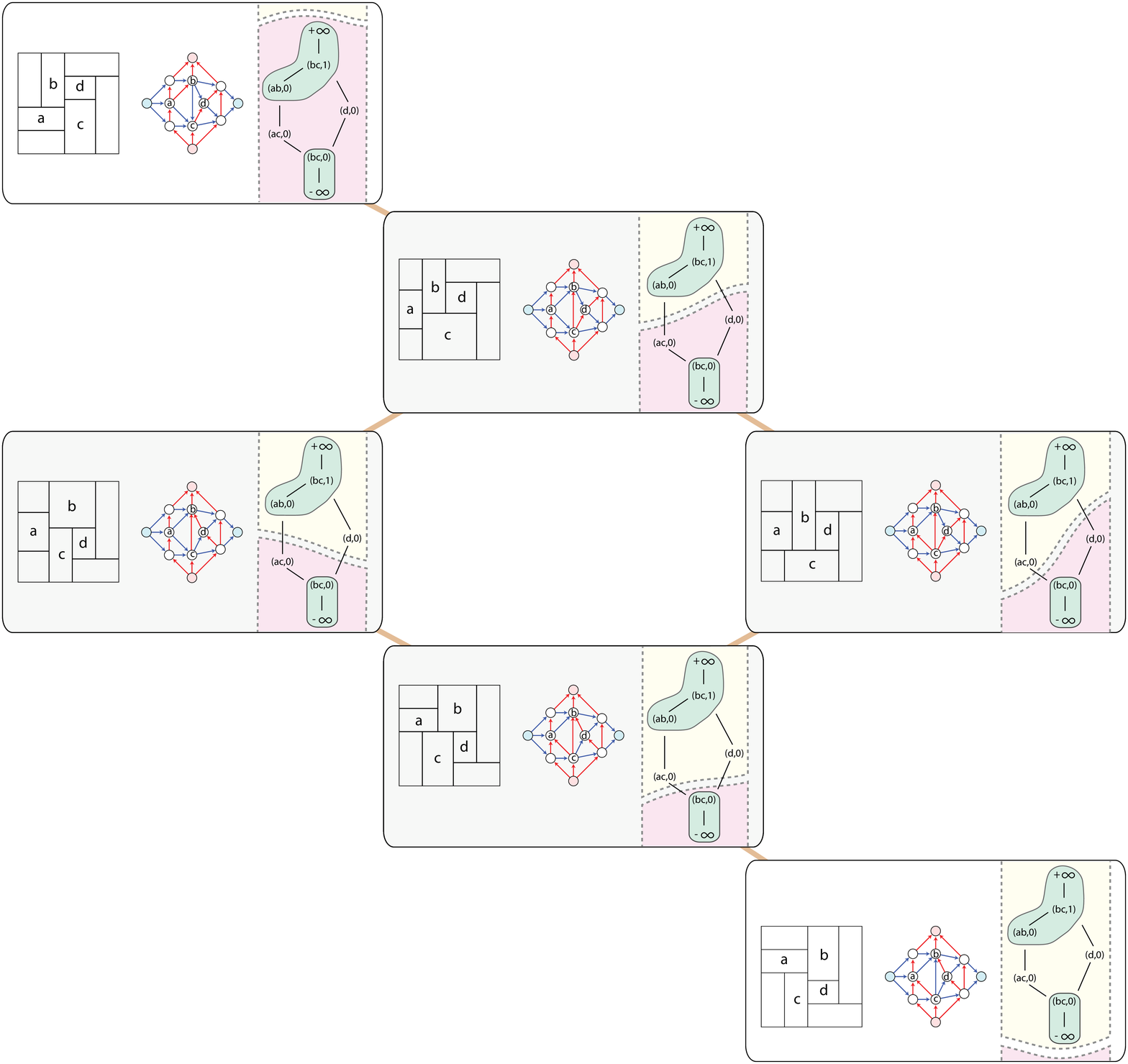}
\vspace{-0.5\baselineskip}
\caption{The family of rectangular layouts dual to a given extended graph $E(\graph)$ satisfying the constraints that the edge between rectangles $a$ and $b$ must be vertical (cannot be colored red) and that the edge between rectangles $b$ and $c$ must be horizontal (cannot be colored blue). The green regions depict strongly connected components of the associated quasiorder $\qorder$. The four central shaded elements of the lattice correspond to layouts satisfying the constraints.}
\label{fig:quasiorder}
\end{figure}

All together, this brings us to the following result:

\begin{lemma}
Let $E(\graph)$ be an extended graph without nontrivial separating 4-cycles and with a given set of forbidden orientations, and let $\qorder$ be the quasiorder formed from the transitive closure of $A(\porder)\cup C$ as described in Lemma~\ref{lem:P-Q}.
Then the elements of $J(\qorder)$ corresponding to partitions of $\qorder$ into two nonempty subsets correspond to exactly the layouts that satisfy the forbidden orientation constraints.
\end{lemma}

\begin{proof}
\label{lem:Q-is-constrained}
By Lemma~\ref{lem:constraint-flip} and the definition of $C$, a partition in $J(\porder)$ corresponds to a constrained layout if and only if it respects each of the edges in $C$. By Lemma~\ref{lem:P-Q}, the elements of $J(\qorder)$ correspond to partitions of $A(\porder)$ that respect $C$. And a partition of $A(\porder)$ corresponds to an element of $J(\porder)$ if and only if its lower set does not contain $\top$ and its upper set does not contain $\bot$.
\end{proof}

\begin{cor}
Let $E(\graph)$ be an extended graph without nontrivial separating 4-cycles and with a given set of forbidden orientations.
There exists a constrained layout for $E(\graph)$ if and only if there exists more than one strongly connected component in $\qorder$.
\end{cor}

\begin{cor}
The existence of a constrained layout for a given extended graph $E(G)$ without nontrivial separating 4-cycles can be proved or disproved in polynomial time.
\end{cor}

\begin{cor}
All constrained layouts for a given extended graph $E(G)$ without nontrivial separating 4-cycles can be listed in polynomial time per layout.
\end{cor}

Figure~\ref{fig:quasiorder} depicts the sublattice resulting from these constructions for the example from Figure~\ref{fig:8layouts}, with constraints on the orientations of two of the layout edges.

\subsection{Junction orientation constraints}

So far we have only considered forbidding certain edge labels. However the method above can easily be extended to different types of constraints. For example, consider two elements of $\porder$  $(x,i)$ and $(y,j)$ that are a covering pair in $\porder$; this implies that $x$ and $y$ are two of the three flippable items surrounding unique a T-junction of the layouts dual to $E(\graph)$. Forcing $(x,i)$ and $(y,j)$ to be equivalent by adding an edge from $(x,i)$ to $(y,j)$ in the constraint graph $C$ can be used for more general constraints: rather than disallowing one or more of the four orientations for any single flippable item, we can disallow one or more of the twelve orientations of any T-junction. For instance, by adding equivalences of this type we could force one of the three rectangles at the T-junction to be the one with the 180-degree angle.

Any internal T-junction of a layout for $E(\graph)$ (dual to a triangle of $\graph$) has 12 potential orientations: each of its three rectangles can be the one with the 180-degree angle, and with that choice fixed there remain four choices for the orientation of the junction. In terms of the {\REL}, any triangle of $\graph$ may be colored and oriented in any of 12 different ways. For a given covering pair $(x,i)$ and $(y,j)$, let $C_{x,y}^{i,j}$ denote the set of edges between pairs $(x,i+4k)$ and $(y,j+4k)$ for all possible integer values of $k$, together with an edge from $\bot$ to $(y,0)$ if $j\bmod 4=0$ and an edge from $(x,i+4k)$ to $\top$ if $i+4k$ is the largest value of $i'$ such that $(x,i')$ belongs to $\porder$. Any T-junction is associated with 12 of these edge sets, as there are three ways of choosing a pair of adjacent flippable items and four ways of choosing values of $i$ and $j$ (mod 4) that lead to covering pairs. Including any one of these edge sets in the constraint graph $C$ corresponds to forbidding one of the 12 potential orientations of the T-junction.

Thus, Lemma~\ref{lem:Q-is-constrained} and its corollaries may be applied without change to dual graphs $E(\graph)$ with junction orientation constraints as well as edge orientation constraints, as long as $E(\graph)$ has no nontrivial separating 4-cycles.

\section{Constrained layouts for unconstrained dual graphs}

\begin{figure}[b]
\centering\includegraphics{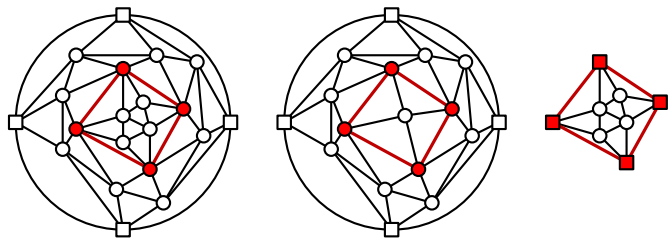}
\vspace{-.5\baselineskip}
\caption{An extended graph with a nontrivial separating four-cycle (left), its outer separation component (center), and its inner separation component (right). From~\cite{area-uni}.}
\label{fig:separation}
\end{figure}

Proper graphs with nontrivial separating 4-cycles still have finite distributive lattices of layouts, but
it is no longer possible to translate orientation constraints into equivalences between members of an underlying partial order. The reason is that, for a graph without trivial separating 4-cycles, the orientation of a feature of the layout changes only for a flip involving that feature, so that the orientation may be determined from the flip count modulo four.  For more general graphs the orientation of a feature is changed not only for flips directly associated with that feature, but also for flips associated with larger 4-cycles that contain the feature, so the flip count of the feature no longer determines its orientation. For this reason, as in\cite{area-uni}, we treat general proper graphs by decomposing them into \emph{minimal separation components} with respect to separating 4-cycles and piecing together solutions found separately within each of these components.

For each separating four-cycle $C$ in a proper graph $\graph$ with a corner assignment $E(\graph)$ consider two minors of $\graph$ defined as follows.
The \emph{inner separation component} of $C$ is a graph $\graph_C$ and its extended graph $E(\graph_C)$, where $\graph_C$ is the subgraph of $\graph$ induced by the vertices inside $C$ and $E(\graph_C)$ adds the four vertices of the cycle as corners of the extended graph. The \emph{outer separation component} of $C$ is a graph formed by contracting the interior of $C$ into a single supervertex. A \emph{minimal separation component} of $\graph$ is a minor of $\graph$ formed by repeatedly splitting larger graphs into separation components until no nontrivial separating four-cycles remain. A partition tree of $E(\graph)$ into minimal separation components may be found in linear time~\cite{area-uni}.

We use the representation of a graph as a tree of minimal separation components in our search for constrained layouts for $\graph$. We first consider each such minimal component separately for every possible mapping of vertices of $C$ to $\{l,t,r,b\}$ (we call these mappings the \emph{orientation} of $E(\graph)$). Different orientations imply different flipping values of forbidden labels for the given constraint function, since the flipping numbers are defined with respect to the orientation of $E(\graph)$.
Having that in mind we are going to test the graph $E(\graph)$ for existence of a constrained layout in the following way:

For each piece in a bottom-up traversal of the decomposition tree and for each
orientation of the corners of the piece:

\begin{enumerate}
\item Find the partial order $\porder$ describing the layouts of the piece

\item Translate the orientation constraints within the piece into a constraint graph on the augmented partial order $A(\porder)$.

\item Compute the strongly connected components of the union of $A(\porder)$ with the constraint graph, and form a binary relation that is a subset of $\qorder$ and that includes all covering relations in $\qorder$ by finding the components containing each pair of elements in each covering relation in $\porder$.

\item Translate the existence or nonexistence of a layout into a constraint on the label of the corresponding degree-4 vertex in the parent piece of the decomposition. That is, if the constrained layout for a given orientation of $E(\graph')$ does not exist, forbid (in the parent piece of the decomposition) the label of the degree-four vertex corresponding to that orientation.
\end{enumerate}

\noindent
If the algorithm above confirms the existence of a constrained layout, we may list all layouts satisfying the constraints as follows. For each piece in the decomposition tree, in top-down order:

\begin{enumerate}
\item List all lower sets of the corresponding quasiorder $\qorder$.

\item Translate each lower set into a layout for that piece.

\item For each layout, and each child of the piece in the
decomposition tree, recursively list the layouts in which the child's corner orientation
matches the labeling of the corresponding degree-four vertex of the outer layout.

\item Glue the inner and outer layouts together.
\end{enumerate}

\begin{theorem}
The existence of a constrained layout for a proper graph $\graph$ can be found in polynomial time in $|\graph|$. The set of all constrained layouts for $graph$ can be found in polynomial time per layout.
\end{theorem}

As described in \cite{area-uni}, the partial order $\porder$ describing the layouts of each piece has a number of elements and covering pairs that is quadratic in the number of vertices in the dual graph of the piece, and a description of this partial order in terms of its covering pairs may be found in quadratic time. The strongly connected component calculation within the algorithm takes time linear in the size of $\porder$, and therefore the overall algorithm for testing the existence of a constrained layout takes time $O(n^2)$, where $n$ is the number of vertices in the given dual graph.

\section{Finding area-universal constrained layouts}
Our previous work \cite{area-uni} included an algorithm for finding area-universal layouts that is fixed-parameter tractable, with the maximum number of separating four-cycles in any piece of the separation component decomposition as its parameter. It is not known whether this problem may be solved in polynomial time for arbitrary graphs. But as we outline in this section, the same fixed parameter tractability result holds for a combination of the constraints from that paper and from this one: the problem of searching for an area-universal layout with constrained orientations.

These layouts correspond to partitions of $\porder$ such that all flippable items that are minimal elements of the upper set and all flippable items that are maximal items of the lower set are all degree-four vertices. A brute force algorithm can find these partitions by looking at all sets of degree-four vertices as candidates for extreme sets for partitions of $\porder$. Instead, in our previous work on this problem we observed that a flippable edge is not \emph{free} in a layout (i.e. cannot be flipped in the layout), if an only if it is \emph{fixed} by a so-called \emph{stretched pair}. A stretched pair is a pair of two degree-four vertices $(v,w)$, such that on any monotone path up from $\layout$ $w$ is flipped before $w$, and on any monotone path down from $\layout$ $v$ is flipped before $w$. If $f_v(\layout)$ is the maximal flipping value of $v$, then we declare $(v,\emptyset)$ to be stretched (where $\emptyset$ is a special symbol) and declare $(\emptyset, w)$ to be a stretched pair if $f_w(\layout) = 0$.  An edge is \emph{fixed} by a stretched pair $(v,w)$ in $\layout$ if $x$ if every monotone path up from $\layout$ moves $w$ before moving $x$, and every monotone path down from $\layout$ moves $v$ before moving $x$. So instead of looking for extreme sets, we could check every set of degree-four vertices for existence of a layout in which every pair in the set is stretched, and check which edges each such a set fixes. Each set $H$ of pairs can be checked for stretchability by starting at the bottom of the lattice and flipping the corresponding items of $\porder$ up the lattice until every pair in $H$ is stretched or the maximal elements of the lattice is reached. If there are $k$ degree-four vertices (or equivalently separating 4-cycles) in a piece, there are $2^{O(k^2)}$ sets of stretched pairs we need consider, each of which takes polynomial time to test, so the overall algorithm for searching for unconstrained area-universal layouts takes time $2^{O(k^2)} n^{O(1)}$.

For constrained layouts, a similar approach works. Within each piece of the separation decomposition, we consider $2^{O(k^2)}$ sets of stretched pairs in $\porder$, as before. However, to test one of these sets, we perform a monotonic sequence of flips in $J(\qorder)$, at each point either flipping an element of $\qorder$ that contains the upper element of a pair that should be stretched, or performing a flip that is a necessary prerequisite to flipping such an upper element. Eventually, this process will either reach an area-universal layout for the piece or the top element of the lattice; in the latter case, no area-universal layout having that pattern of stretched pairs exists. By testing all sets of stretched pairs, we may find whether an area-universal layout matching the constraints exists for any corner coloring of any piece in the separation decomposition. These constrained layouts for individual pieces can then be combined by the same tree traversal of the separation decomposition tree that we used in the previous section, due to the observation from \cite{area-uni} that a layout is area-universal if and only if the derived layout within each of its separation components is area-universal. The running time for this fixed-parameter tractable algorithm is the same as in~\cite{area-uni}.

\section{Conclusions and open problems}\label{sec:disc}

We have provided efficient algorithms for finding rectangular layouts with orientation constraints on the features of the constraints, and we have outlined how to combine our approach with the previous algorithms for finding area-universal layouts so that we can find orientation-constrained area-universal layouts as efficiently as we can solve the unconstrained problem.

An important problem in the generation of rectangular layouts with special properties, that has resisted our lattice-theoretic approach, is the generation of sliceable layouts. If we are given a graph $\graph$, can we determine whether it is the graph of a sliceable layout in polynomial time? Additionally, although our algorithms are polynomial time, there seems no reason intrinsic to the problem for them to take as much time as they do: can we achieve subquadratic time bounds for finding orientation-constrained layouts, perhaps by using an algorithm based more on the special features of the problem and less on general ideas from lattice theory?

Moving beyond layouts, there are several other important combinatorial constructions that may be represented using finite distributive lattices, notably the set of matchings and the set of spanning trees of a planar graph, and certain sets of orientations of arbitrary graphs~\cite{propp93}. It would be of interest to investigate whether our approach of combining the underlying partial order of a lattice with a constraint graph produces useful versions of constrained matching and constrained spanning tree problems, and whether other algorithms that have been developed in the more general context of distributive finite lattices~\cite{propp97} might fruitfully be applied to lattices of rectangular layouts.

\section*{Acknowledgements}

Work of D. Eppstein was supported in part by NSF grant
0830403 and by the Office of Naval Research under grant
N00014-08-1-1015.
{\small\raggedright
\bibliographystyle{abbrv}
\bibliography{rectilinear}}

%
\end{document}